\pdfoutput=1
\documentclass[11pt]{article}
\usepackage{acl}
\usepackage{times}
\usepackage{latexsym}
\usepackage[T1]{fontenc}
\usepackage[utf8]{inputenc}
\usepackage{microtype}

\usepackage{amsfonts}
\usepackage{amsmath}
\usepackage{amssymb}
\usepackage{amsthm}
\usepackage{booktabs}
\usepackage{dirtytalk}
\usepackage{graphicx}
\usepackage{natbib}

\theoremstyle{definition}
\newtheorem{definition}{Definition}[section]
\newtheorem{remark}{Remark}[section]
\newtheorem{lemma}{Lemma}[section]
\newtheorem{theorem}{Theorem}[section]

\newcommand{\One}{\ensuremath{\bar{1}}}
\newcommand{\Zero}{\ensuremath{\bar{0}}}
\newcommand{\State}{\ensuremath{Q}}
\newcommand{\Alphabet}{\ensuremath{\Sigma}}
\newcommand{\Weight}{\ensuremath{\mathbb{K}}}

\newcommand{\astar}{A*}
\newcommand{\real}{\ensuremath{\mathbb{R}}}
\newcommand{\preal}{\ensuremath{\mathbb{R}_+}}

\renewcommand{\max}{\ensuremath{\textrm{max}}}
\renewcommand{\min}{\ensuremath{\textrm{min}}}
\renewcommand{\ln}{\ensuremath{\textrm{ln}}}
\renewcommand{\log}{\ensuremath{\textrm{log}}}

\newcommand{\set}[1]{\left\{#1\right\}}

\title{\astar~shortest string decoding for non-idempotent semirings}
\author{Kyle Gorman \& Cyril Allauzen \\ Google LLC}
\date{}

\begin{document}

\maketitle

\begin{abstract}%
The single shortest path algorithm is undefined for weighted finite-state automata over non-idempotent semirings because such semirings do not guarantee the existence of a shortest path.
However, in non-idempotent semirings admitting an order satisfying a monotonicity condition (such as the plus-times or log semirings),
the shortest string is well-defined.
We describe an algorithm which finds the shortest string for a weighted non-deterministic automaton over such semirings using the backwards shortest distance of an
equivalent deterministic automaton (DFA) as a heuristic for \astar~search performed over a companion idempotent semiring,
This algorithm is proven to return the shortest string.
There may be exponentially more states in the equivalent DFA, but the proposed algorithm needs to visit only a small fraction of them if determinization is performed \say{on the fly}.
\end{abstract}

\section{Introduction}

Weighted finite-state automata provide a compact representation of hypotheses in various speech recognition and text processing applications \citep[e.g.,][]{Mohri1997b,MohriEtAl2002,RoarkSproat2007,FSTP}.
Under a wide range of assumptions,
weighted finite-state lattices allow for efficient polynomial-time decoding via shortest-path algorithms \citep{Mohri2002}.

The shortest path%
---and the algorithms that compute it---%
are well-defined when the weights of a lattice are \emph{idempotent} and exhibit the \emph{path property}.
These properties are formalized below, but
informally they hold that the distance between any two states corresponds to a single path between those states,
so that the shortest-path algorithm%
---having identified this path---%
does not need to consider the weights of competing paths between those states.
However, when the weights of a lattice lack these two properties,
there is no guarantee that a shortest path between any two states exists.
This situation arises in many speech and language technologies.
For instance, generative models for speech recognition and machine translation%
---and in many unsupervised settings---%
often use \emph{expectation maximization} \citep[EM;][]{DempsterEtAl1977} or related algorithms for learning; such models generally lack these two key properties.
Under many conditions, efficient decoding of a lattice constructed using EM is required;
in this case,
one can decode approximately by interpreting the lattice as if it were idempotent and had the path property,
or one can train the model using the Viterbi approximation to EM,
and then decode using an ordinary shortest-path algorithm.%
\footnote{
    Both of these strategies are discussed by \citeauthor{BrownEtAl1993} (\citeyear{BrownEtAl1993}; see \S4.3 and \S6.2 respectively), though they do not refer to these semiring properties by name.
}

In non-idempotent semirings admitting an order satisfying a monotonicity condition, the shortest path is undefined but the closely related notion of \emph{shortest string} is well-defined.
We show below that it is still possible to efficiently determine the shortest string for lattices defined over non-idempotent monotonic negative semirings such as the plus-times
and log semirings, both used for expectation maximization.
We propose a simple algorithm for decoding the shortest string over such semirings which combines shortest-path search with the \astar~queue discipline \citep{HartEtAl1968} and \say{on the fly} determinization \citep{Mohri1997b}.
After providing definitions and the algorithm,
we describe an implementation and evaluate it using word lattices produced by a speech recognizer.
The algorithm%
---in contrast to a naïve algorithm---%
is shown to scale well as a function of lattice size.

\section{Definitions}

Before we introduce the proposed decoding algorithm we provide definitions of key notions.

\subsection{Semirings}

Weighted automata algorithms operate with respect to an algebraic system known as a \emph{semiring}, characterized by the combination of two \emph{monoids}.

\begin{definition}
A \emph{monoid} is a pair $(\Weight, \bullet)$
where $\Weight$ is a set
and $\bullet$ is a binary operator over $\Weight$
with the following properties:
\begin{enumerate}
\item \emph{closure}: $\forall a, b \in \Weight : a \bullet b \in \Weight$.
\item \emph{associativity}: $\forall a, b, c \in \Weight : (a \bullet b) \bullet c = a \bullet (b \bullet c)$.
\item \emph{identity}: there exists an identity element $e \in \Weight$ such that $\forall a \in \Weight : e \bullet a = a \bullet e = a$.
\end{enumerate}
\end{definition}

\begin{definition}
A monoid is \emph{commutative} in the case that $\forall a, b \in \Weight : a \bullet b = b \bullet a$.
\end{definition}

\begin{definition}
\label{d:semiring}
A semiring is a five-tuple $(\Weight, \oplus, \otimes, \Zero, \One)$ where:
\begin{enumerate}
\item $(\Weight, \oplus)$ is a commutative monoid with identity element $\Zero$.
\item $(\Weight, \otimes)$ is a monoid with identity element $\One$.
\item $\forall a \in \Weight : a \otimes \Zero = \Zero \otimes a = \Zero$.
\item $\forall a, b, c \in \Weight : a \otimes (b \oplus c) = (a \otimes b) \oplus (a \otimes c)$.
\end{enumerate}
\end{definition}

\begin{definition}
A semiring is \emph{zero-sum-free} if non-\Zero~elements cannot sum to \Zero;
that is, $\forall a, b \in \Weight : a \oplus b = \Zero \implies a = b = \Zero$.
\end{definition}

\begin{definition}
A semiring is \emph{idempotent} if $\oplus$ is idempotent;
that is, $\forall a \in \Weight : a \oplus a = a$.
\label{d:idempotent}
\end{definition}

\begin{definition}
A semiring has the \emph{path property} if  $\forall a, b \in \Weight : a \oplus b \in \{a, b\}$.
\label{d:path}
\end{definition}

\begin{remark}
If a semiring has the path property it is also idempotent.
\end{remark}

\begin{definition}
The \emph{natural order} of an idempotent semiring is a boolean operator $\preceq$
such that $\forall a, b \in \Weight : a \preceq b$ if $a \oplus b = a$.
\end{definition}

\begin{remark}
In a semiring with the path property, the natural order is a \emph{total} order.
That is, $\forall a, b \in \Weight$, either $a \preceq b$ or $b \preceq a$.
\end{remark}

\begin{definition}
A semiring is \emph{monotonic} if $\forall a, b, c \in \Weight$, $ a \preceq b$ implies:
\begin{enumerate}
\item $a \oplus c \preceq b \oplus c$.
\item $a \otimes c \preceq b \otimes c$.
\item $c \otimes a \preceq c \otimes b$.
\end{enumerate}
\end{definition}

\begin{definition}
A semiring is \emph{negative} if $\One \preceq \Zero$.
\end{definition}

\begin{remark}
In a monotonic negative semiring, $\forall a, b \in \Weight : a \preceq \Zero$ 
and $a \oplus b \preceq b$.
\end{remark}

Some examples of monotonic negative semirings are given in \autoref{t:semirings}.

\begin{table*}[ht]
\centering
\begin{tabular}{l r cccccc}
\toprule
            && $\Weight$                         & $\oplus$        & $\otimes$ & $\Zero$   & $\One$ & $\preceq$ \\
\midrule
Plus-times  && $\preal$                          & $+$             & $\times$  & $0$       & $1$    & $\geq$ \\
Max-times   && $\preal$                          & $\max$          & $\times$  & $0$       & $1$    & $\geq$ \\ 
Log         && $\real \cup \{-\infty, +\infty\}$ & $\oplus_{\log}$ & $+$       & $+\infty$ & $0$    & $\leq$ \\
Tropical    && $\real \cup \{-\infty, +\infty\}$ & $\min$          & $+$       & $+\infty$ & $0$    & $\leq$ \\
\bottomrule
\end{tabular}
\caption{Common monotonic negative semirings and the associated natural orders; $a \oplus_{\log} b = -\ln(e^{-a}+ e^{-b})$.}
\label{t:semirings}
\end{table*}

\begin{definition}
The \emph{companion semiring} of a monotonic negative semiring $(\Weight, \oplus, \otimes, \Zero, \One)$ with total order $\preceq$ is the semiring $(\Weight, \widehat{\oplus}, \otimes, \Zero, \One)$ where $\widehat{\oplus}$ is the minimum binary operator for $\preceq$:
\begin{equation*}
    a~\widehat{\oplus}~b = \begin{cases}
        a & \textrm{if } a \preceq b \\
        b & \textrm{otherwise} .\\
    \end{cases}
\end{equation*}
\end{definition}

\begin{remark}
The max-times and tropical semirings are companion semirings to the plus-times
and log semirings, respectively.
\end{remark}

\begin{remark}
By construction a companion semiring has the path property and natural order $\preceq$.
\end{remark}

\subsection{Weighted finite-state acceptors}

Without loss of generality,
we consider single-source $\epsilon$-free weighted finite-state acceptors.%
\footnote{
    The definition provided here can easily be generalized to automata with multiple initial states,
    a single final state,
    initial or final weights,
    or $\epsilon$-transitions
    (e.g., \citealp[~ch.~1]{RoarkSproat2007};
    \citealp{Mohri2009};
    \citealp[~ch.~1]{FSTP}).
}

\begin{definition}
A \emph{weighted finite-state acceptor} (WFSA) is defined by a five-tuple
$(Q, s, \Alphabet, \omega, \delta)$
and a semiring $(\Weight, \oplus, \otimes, \Zero, \One)$
where:

\begin{enumerate}
\item $\State$ is a finite set of states.
\item $s \in \State$ is the \emph{initial state}.
\item $\Alphabet$ is the \emph{alphabet}.
\item $\omega \subseteq \State \times \Weight$ is the \emph{final weight function}.
\item $\delta \subseteq \State \times \Alphabet \times \Weight \times \State$ is the \emph{transition relation}.
\end{enumerate}
\end{definition}

\begin{definition}
An WFSA is \emph{acyclic} if there exists a \emph{topological ordering},
an ordering of the states such that if there is a transition from state $q$ to $r$
where $q, r \in \State$, then $q$ is ordered before $r$.
Otherwise, the WFSA is \emph{cyclic}.
\end{definition}

\subsection{Shortest distance}

\begin{definition}
A state $q \in \State$ is \emph{final} if $\omega(q) \ne \Zero$.
\end{definition}

\begin{definition}
Let $F = \{q \mid \omega(q) \ne \Zero\}$ denote the set of final states.
\end{definition}

\begin{definition}
A \emph{path} through an acceptor $p$ is a triple consisting of:

\begin{enumerate}
\item a {state sequence} $q[p] = q_1, q_2, \ldots, q_n \in \State^n$,
\item a {weight sequence} $k[p] = k_1, k_2, \ldots, k_n \in \Weight^n$, and
\item a {string} $z[p] = z_1, z_2 \ldots, z_n \in \Alphabet^{n}$
\end{enumerate}
such that $\forall i \in [1, n] : (q_i, z_i, k_i, q_{i + 1}) \in \delta$;
that is, each transition from $q_i$ to $q_{i + 1}$ must have label $z_i$ and weight $k_i$.
\end{definition}

\begin{definition}
Let $P_{q \rightarrow r}$ be the set of all paths from $q$ to $r$ where $q, r \in \State$.
\end{definition}

\begin{definition}
The \emph{forward shortest distance} $\alpha \subseteq \State \times \Weight$ is a partial function from a state $q \in \State$ that gives the $\oplus$-sum of the $\otimes$-product of the weights of all paths from the initial state $s$ to $q$:
\begin{equation*}
\alpha(q) = \bigoplus_{
p \in P_{s \rightarrow q}} \bigotimes_{k_i \in k[p]} k_i  .
\end{equation*}
\end{definition}

\begin{definition}
The \emph{backwards shortest distance} $\beta \subseteq \State \times \Weight$ is a partial function from a state $q \in \State$  that gives the $\oplus$-sum of the $\otimes$-product of the weights of all paths from
$q$ to a final state, including the final weight of that final state:

\begin{equation*}
\beta(q) = \bigoplus_{f \in F} \left( \bigoplus_{p \in P_{q \rightarrow f}} \bigotimes_{k_i \in k[p]} k_i \otimes \omega(f) \right).
\end{equation*}
\end{definition}

\begin{definition}
A state is \emph{accessible} if there exists a path to it from the initial state $s$.
\end{definition}

\begin{definition}
A state is \emph{coaccessible} if there exists a path from it to a final state $f \in F$.
\end{definition}

\begin{remark}
For a state $q$, $\alpha(q)$ and $\beta(q)$ are defined if and only if $q$ is accessible and coaccessible, respectively.
\end{remark}

\begin{definition}
The \emph{total shortest distance} of an automaton is $\beta(s)$.
\end{definition}

\subsection{Shortest path}

\begin{definition}
A path is \emph{complete} if
\begin{enumerate}
\item $(s, z_1, k_1, q_1) \in \delta$.
\item $q_n \in F$.
\end{enumerate}
\noindent
That is, a complete path must also begin with an arc from the initial state $s$ to $q_1$ with label $z_1$ and weight $k_1$, and halt in a final state.
\end{definition}

\begin{definition}
The weight of a complete path is given by the $\otimes$-product of its weight sequence and its final weight:
\begin{equation*}
\bar{k} = \left(\bigotimes_{k_i \in k[p]} k_i \right) \otimes \omega(q_n).
\end{equation*}
\end{definition}

\begin{definition}
A \emph{shortest path} through an automaton is
a complete path whose weight is equal to the total shortest distance $\beta(s)$.
\end{definition}

\begin{remark}
Automata over non-idempotent semirings may lack a shortest path \citep[322]{Mohri2002}.
Consider for example the NFA shown in the left side of \autoref{f:motivating}.
Let us assume that $k \oplus k \preceq k < k'$.
Then, the total shortest distance is $k \oplus k$ but the shortest path is $k$.
By definition, a non-idempotent semiring does not guarantee that these two weights will be equal.
Then there is no complete path whose weight is that of the total shortest distance, and thus no shortest path exists.
\end{remark}

\begin{remark}
It is not possible in general to efficiently find the shortest path over non-monotonic semirings.
See \citet{Mohri2002} for general conditions under which the shortest path can be found in polynomial time.

\end{remark}

\subsection{Determinization}

\begin{definition}
A WFSA is \emph{deterministic} if,
for each state $q \in \State$,
there is at most one transition with a given label $z \in \Sigma$ from that state,
and \emph{non-deterministic} otherwise.
\end{definition}

\begin{definition}
A zero-sum-free semiring is \emph{weakly divisible} if
\begin{equation*}
\forall a,b \in \Weight \; \exists c \in \Weight : a = (a \oplus b) \otimes c.
\end{equation*}
\end{definition}

\begin{definition}
A weakly divisible semiring is \emph{cancellative} if $c$ is unique and can thus be denoted by $c = (a \oplus b)^{-1} a$ \citep[238]{Mohri2009}.
\end{definition}

\begin{remark}
All semirings in \autoref{t:semirings} are zero-sum-free, weakly divisible, and cancellative.
\end{remark}

\begin{remark}
\label{r:nfa}
For every non-deterministic, acyclic WFSA (or NFA) over a zero-sum-free, weakly divisible and cancellative semiring,
there exists an equivalent deterministic WFSA (or DFA).
However, a DFA may be exponentially larger than an equivalent NFA \citep[\S2.3.6]{HopcroftEtAl2008}.
\end{remark}

We now provide a brief presentation of the determinization algorithm for WFSAs.
Proofs can be found in \citealt{Mohri1997b}.
Given an WFSA $A = (Q, s, \Alphabet, \omega, \delta)$ over a zero-sum-free, weakly divisible and cancellative semiring $(\Weight, \oplus, \otimes, \Zero, \One)$,
its equivalent DFA can be defined and constructed as the DFA  $A_d = (Q_d, s_d, \Alphabet, \omega_d, \delta_d)$ where
${\State}_d$ is a finite set whose elements are subsets of $\State \times \Weight$, recursively defined as follows:
\begin{enumerate}
\item $s_d = \set{(s, \One)} \in {\State}_d$.
\item $\kappa_d \subseteq \State_d \times \Alphabet \times \Weight$ is the \emph{weight transition function}, defined as
\begin{equation*}
\kappa_d(q, z) = \bigoplus_{(q_i, k_i) \in q} k_i \otimes \left( \bigoplus_{(q_i, z, k_j, r_j) \in \delta} k_j \right).
\end{equation*}
\item $\nu_d \subseteq \State_d \times \Alphabet \times \State_d$ is the \emph{next-state transition function}, defined as $\nu_d(q, z) =$
\begin{equation*}
\bigcup_{\footnotesize\begin{array}{r}{} (q_i, k_i) \in q \\ (q_i, z, k_j, r_j) \in \delta \end{array}} \set{(r_j, \kappa_d(q, z)^{-1} l_j)}
\end{equation*}
where $l_j = \bigoplus_{(q_i, z, k_j, r_j) \in \delta} k_i \otimes k_j$.
\item $Q_d = \nu_d^*(s_d, \Alphabet)$ defines the set of states as the closure of the next-state transition function.
\end{enumerate}
The transition relation is then defined as 
\begin{equation*}
\delta_d = \set{(q, z , \kappa_d(q, z), \nu_q(q,z)) | (q, z) \in \State_d \times \Alphabet}
\end{equation*}
and the final weight function $\omega_d \subseteq Q_d \times \Weight$ as
\begin{equation*}
\omega_d(q) =  \bigoplus_{(q_i, k_i) \in q} k_i \otimes \omega(q_i).
\end{equation*}
The intuition underlying this construction is that a state $q\in Q_d$ encodes a set of states in $Q$ that can be reached from $s$ by some common strings.
More precisely, let $p'$ be the unique path in $P_{s_d \rightarrow q}$ labeled by some $z' \in \Alphabet^*$, then for any $(q_i, k_i) \in q$:
\begin{equation*}
k[p'] \otimes k_i = \bigoplus_{p \in P_{s \rightarrow q_i} : z[p] = z'} k[p].
\end{equation*}
Termination is guaranteed for acyclic WFSAs \citep{Mohri1997b}.

\autoref{f:motivating} gives an example of an NFA and an equivalent DFA. States 0 and 1 in the DFA correspond respectively to the subsets ${(0, \One)}$ and ${(1, \One)}$ and $\kappa_d(0, a) = k \otimes k$.

\begin{figure*}[ht]
\centering
\begin{tabular}{l c l}
\includegraphics[width=.3\linewidth]{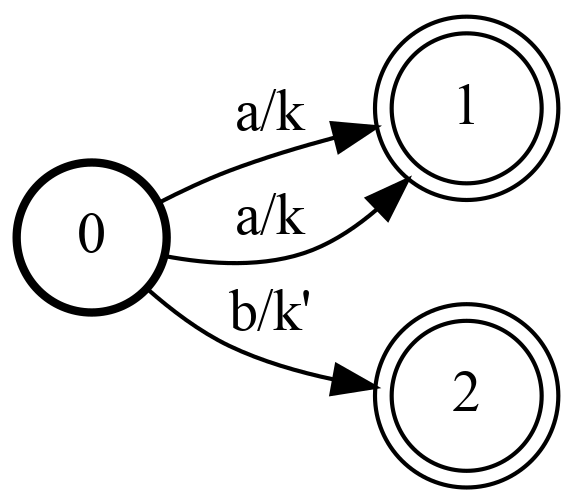} && \includegraphics[width=.31\linewidth]{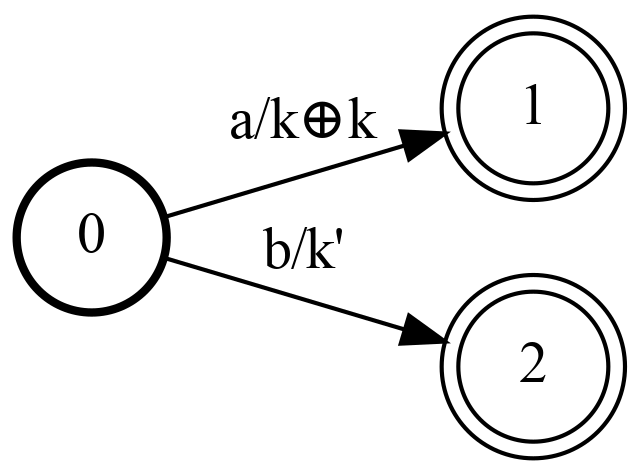}
\end{tabular}
\caption{
State diagrams showing a weighted NFA (left) and an equivalent DFA (right).
}
\label{f:motivating}
\end{figure*}

\begin{remark}
Given a NFA $A$ with backwards shortest distance $\beta$, the backwards shortest distance $\beta_d$ over the equivalent DFA $A_d$ can be computed from $\beta$:
\begin{equation*}
\beta_d(q) = \bigoplus_{(q_i, k_i) \in q} k_i \otimes \beta(q_i)
\end{equation*}
\label{r:det-sd}
for any $q \in Q_d$ \citep{MohriRiley2002}.
\end{remark}

Since $A$ is assumed to be acyclic, $\beta$ can be computed in $O(|Q|)$ time \citep[\S4.1]{Mohri2002},
and once $\beta$ has been computed, $\beta_d(q)$ can also be computed in linear time in $|q| \le |Q|$ for any $q \in Q_d$.
This computation can be performed lazily (\say{on the fly}) as soon as the existence of $q \in Q_d$ is known,
without requiring $A_d$ to be fully constructed.

\subsection{Shortest string}

\begin{definition}
Let $P_z$ be the set of paths with string $z \in \Sigma^*$,
and let the weight of $P_z$ be

\begin{equation*}
\sigma(z) = \bigoplus_{p \in P_z} \bar{k}[p] .
\end{equation*}
\end{definition}

\begin{definition}
A \emph{shortest string} $z$ is one such that $\forall z' \in \Sigma^*, \sigma(z) \preceq \sigma(z')$.
\end{definition}

\begin{lemma}
In an idempotent semiring, a shortest path's string is also a shortest string.
\end{lemma}

\begin{proof}
Let $p$ be a shortest path. By definition, $\bar{k}[p] \preceq \bar{k}[p']$ for all complete paths $p'$. It follows that $\forall z' \in \Sigma^*$
\begin{align*}
\sigma(z[p]) &= \bigoplus_{p \in P_z} \bar{k}[p] \preceq \sigma(z'[p']) \\
             &=  \bigoplus_{p' \in P_z} \bar{k}[p']
\end{align*}

\noindent
thus $z[p]$ is the shortest string.
\end{proof}

\begin{lemma}
\label{lem:shortest-string-det}
In a DFA over a monotonic semiring, a shortest string is the string of a shortest path in that DFA viewed as an WFSA over the corresponding companion semiring.
\end{lemma}

\begin{proof}
Determinism implies that for all complete path $p'$, $\bar{k}[p'] = \sigma(z[p'])$. 
Let $z$ be the shortest string in the DFA and $p$ the unique path admitting the string $z$.
Then

\begin{equation*}
\bar{k}[p] = \sigma(z) \preceq \sigma(z[p']) = \bar{k}[p']
\end{equation*}

\noindent
for any complete path $p'$.
Hence

\begin{equation*}
\bar{k}[p] = \widehat{\bigoplus_{p' \in P_{s \rightarrow F}}} \bar{k}[p'] . 
\end{equation*}

\noindent
Thus $p$ is a shortest path in the DFA viewed over the companion semiring.
\end{proof}

\subsection{\astar~search}

\astar~search \citep{HartEtAl1968} is a common \emph{shortest-first} search strategy for computing the shortest path in a WFSA over an idempotent semiring.
It can be thought of as a variant of \citeauthor{Dijkstra1959}'s (\citeyear{Dijkstra1959}) algorithm, in which exploration is guided by a shortest-first priority queue discipline.

In Dijkstra's algorithm, at every iteration the algorithm explores the state $q$ which minimizes $\alpha(q)$, the shortest distance from the initial state $s$ to $q$, until all states have been visited.
In \astar~search, search priority is determined by a some function of $\digamma \subseteq \State \times \Weight$, known as the \emph{heuristic}, which gives an estimate of the weight of paths from some state to a final state.
At every iteration, \astar~instead explores the state $q$ which minimizes $\alpha(q) \otimes \digamma(q)$.%
\footnote{
    One can thus view \citeauthor{Dijkstra1959}'s algorithm as a special case of \astar~search with the uninformative heuristic $\digamma = \One$.
}

\begin{definition}
An \astar~heuristic is \emph{admissible} if it never overestimates the shortest distance to a state \citep[103]{HartEtAl1968}.
That is, it is admissible if $\forall q \in Q : \digamma(q) \preceq \beta(q)$.
\end{definition}

\begin{definition}
An \astar~heuristic is \emph{consistent} 
if it never overestimates the cost of reaching a successor state.
That is, it is consistent if $\forall q, r \in Q$ such that $\digamma(q) \preceq k \otimes \digamma(r)$ if $(q, z, k, r) \in \delta$,
i.e., if there is a transition from $q$ to $r$ with some label $z$ and weight $k$.
\end{definition}


\begin{remark}
\label{rem:astar}
If $\digamma$ is \emph{admissible} and \emph{consistent},
\astar~search is guaranteed to find a shortest path (if one exists) after visiting all states such that $\digamma[q] \preceq \beta[s]$ \citep[104f.]{HartEtAl1968}.
\end{remark}

\section{The algorithm}
\label{s:algorithm}

Consider an acyclic, $\epsilon$-free WFSA over a monotonic negative semiring $(\Weight, \oplus, \otimes, \Zero, \One)$ with total order $\preceq$ for which we wish to find the shortest string.
The same WFSA can also be viewed as a WFSA over the corresponding companion semiring $(\Weight, \widehat{\oplus}, \otimes, \Zero, \One)$,
and we denote by $\widehat{\beta}$ the backward shortest-distance over this companion semiring.
We prove two theorems, and then introduce an algorithm for search.

\begin{theorem}
The backwards shortest distance of an WFSA over a monotonic negative semiring is an admissible heuristic for the \astar~search over its companion semiring.
\end{theorem}

\begin{proof}
In a monotonic negative semiring, the $\oplus$-sum of any $n$ terms is upper-bounded by each of the $n$ terms and hence by the $\widehat{\oplus}$-sum of these $n$ terms. It follows that
\begin{align*}
    \digamma(q) &= \beta(q)  \\
                &=\bigoplus_{p \in P_{q \rightarrow F}} \bar{k}[p]
                \preceq \widehat{\bigoplus_{p \in P_{q \rightarrow F}}} \bar{k}[p]  \\
                &= \widehat{\beta}(q),
\end{align*}

\noindent
and this shows that $\digamma = \beta$ is an admissible heuristic for $\widehat{\beta}$.
\end{proof}

\begin{theorem}
The backwards shortest distance of an WFSA over a monotonic negative semiring is a consistent heuristic for the \astar~search over its companion semiring.
\end{theorem}

\begin{proof}
Let $(q, z, k, r)$ be a transition in $\delta$. Leveraging again the property that an $\oplus$-sum of any $n$ terms is upper-bounded by any of these terms, we show that
\begin{align*}
 \digamma(q) &= \beta(q) \\
 &= \bigoplus_{p \in P_{q \rightarrow F}} \bar{k}[p] \\
 &= \bigoplus_{(q, z', k', r') \in \delta} k' \otimes \beta(r') \preceq k \otimes \beta(r) \\
 &= k \otimes \digamma(r)
\end{align*}

\noindent
showing $\digamma = \beta$ is a consistent heuristic.
\end{proof}

Having established that this is an admissible and consistent heuristic for \astar~search over the companion semiring,
a naïve algorithm then suggests itself,
following Lemma~\ref{lem:shortest-string-det} and Remark~\ref{rem:astar}.
Given a non-deterministic WFSA over the monotonic negative semiring $(\Weight, \oplus, \otimes, \Zero, \One)$, apply determinization to obtain an equivalent DFA,
compute $\beta_d$, the backwards shortest distance over the resulting DFA over $(\Weight, \oplus, \otimes, \Zero, \One)$ and then perform \astar~search over the companion semiring using $\beta_d$ as the heuristic.
However, as mentioned in Remark \ref{r:nfa},
determinization has an exponential worse-case complexity in time and space and is often prohibitive in practice.
Yet determinization%
---and the computation of elements of $\beta_d$---%
only need to be performed for states actually visited by \astar~search.
Let $\beta$ denote the backwards shortest distance over a non-deterministic WFSA over the monotonic negative semiring $(\Weight, \oplus, \otimes, \Zero, \One)$.
Then, the algorithm is as follows:

\begin{enumerate}
\item Compute $\beta$ over $(\Weight, \oplus, \otimes, \Zero, \One)$.
\item Lazily determinize the WFSA, lazily computing $\beta_d$ from $\beta$ over $(\Weight, \oplus, \otimes, \Zero, \One)$.
\item Perform \astar~search for the shortest string over $(\Weight, \widehat{\oplus}, \otimes, \Zero, \One)$ with $\beta_d$ as the heuristic.
\end{enumerate}

\section{Evaluation}
\label{s:evaluation}

We evaluate the proposed algorithm using non-idempotent speech recognition lattices.

\subsection{Data}

We search for the shortest string in a sample of 700 word lattices derived from Google Voice Search traffic.
This data set was previously used by \citet{MohriRiley2015} and \citet[ch.~4]{FSTP} for evaluating related WFSA algorithms.
Each path in these lattices is a single hypothesis transcription produced by a production-grade automatic speech recognizer, here treated as a black box.
The exact size of each input lattice size is determined by a probability threshold,
so paths with probabilities below a certain threshold have been pruned.
These lattices are acyclic, $\epsilon$-free, non-deterministic WFSAs over the log semiring, a monotonic non-idempotent semiring.

\subsection{Implementation}

The above algorithm is implemented as part of 
OpenGrm-BaumWelch,
an open-source C++17 library released under the Apache-2.0 license.%
\footnote{
   \url{https://baumwelch.opengrm.org}
}
This toolkit includes \texttt{baumwelchdecode},
a command-line tool which implements the above algorithm over the log semiring,
using the tropical semiring as a companion semiring.
This implementation depends in turn on implementations of determinization, shortest distance, and shortest path algorithms provided by OpenFst \citep{AllauzenEtAl2007}.
This command-line tool, along with various OpenFst command-line utilities, were used to conduct the following experiment.

\subsection{Methods}

We compare the proposed algorithm to the naïve algorithm mentioned in (\S\ref{s:algorithm}).
The naïve algorithm first exhaustively constructs the equivalent DFA by applying weighted determinization%
---as implemented by OpenFst's \texttt{fstdeterminize} command-line tool---%
then performs \astar~search on the DFA over the companion semiring.
Its complexity is bounded by the number of states in the full DFA.
In contrast, the complexity of the proposed algorithm is bounded by the number of DFA states dynamically constructed%
---i.e., when they are visited---%
during search.
As an additional measure, we also compare the number of states visited by the proposed algorithm to the number of states in the original NFA lattice.

\subsection{Results}

\autoref{f:dfa} compares the proposed algorithm to the naïve algorithm.
One can see that the naïve algorithm may in some cases have to construct upwards of 100,000 states for word lattices where the proposed algorithm need only construct hundreds of states.
This demonstrates that the proposed algorithm is substantially more efficient than the naïve algorithm.
\autoref{f:nfa} visualizes the number of states visited by the proposed algorithm as a function of the size of the input NFA.

\begin{figure}[ht]
\centering
\includegraphics{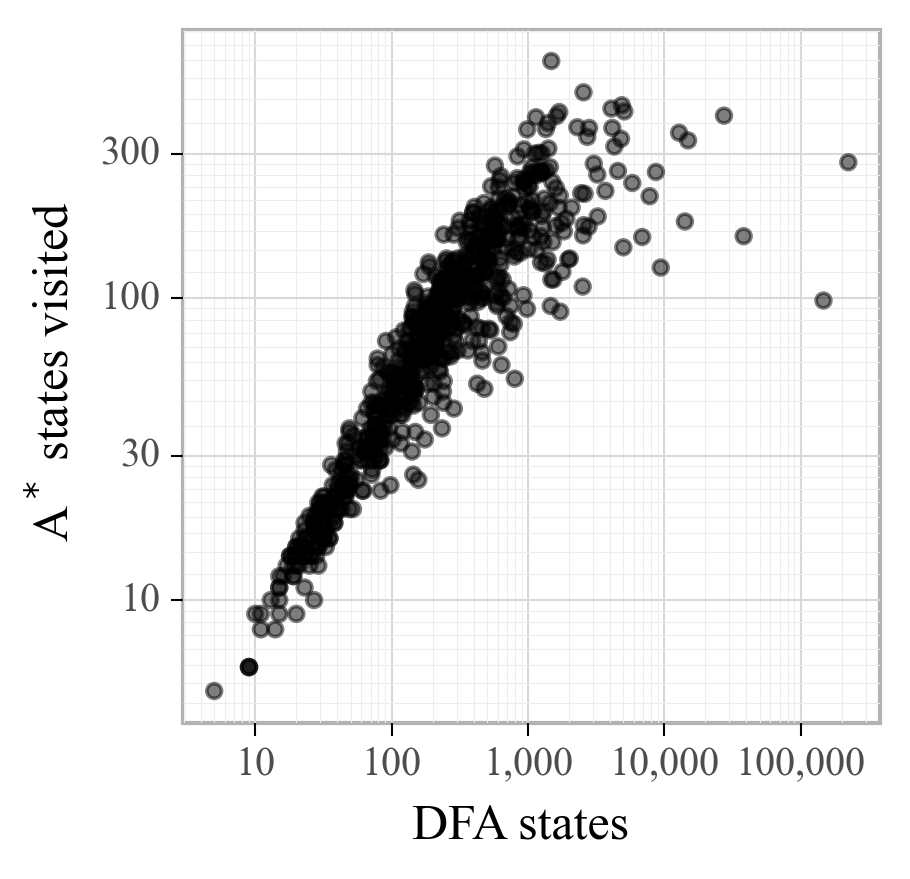}
\caption{
Comparison of word lattice decoding 
with the proposed algorithm vs.~the naïve algorithm.
The $x$-axis shows the number of states in the full DFA; the $y$-axis shows the number of DFA states visited by the proposed algorithm.
Both axes are in logarithmic scale.}
\label{f:dfa}
\end{figure}

\begin{figure}[ht]
\centering
\includegraphics{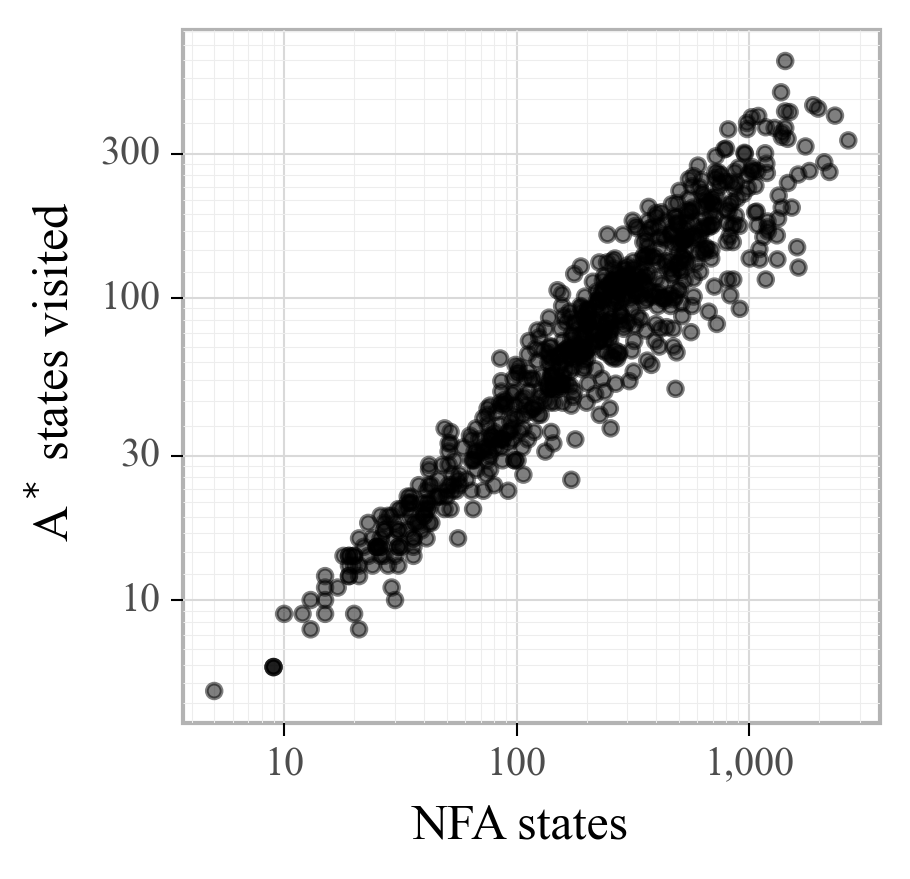}
\caption{
Comparison of word lattice decoding with the proposed algorithm
to the size of the input NFA.
The $x$-axis shows the number of states in the input NFA; the $y$-axis shows the number of states visited by the proposed algorithm.
Both axes are in logarithmic scale.}
\label{f:nfa}
\end{figure}

\section{Related work}

Several prior studies use \astar~search for decoding speech lattices over idempotent semirings.
For example, \citet{MohriRiley2002} describe a related algorithm for computing $n$-best lists over an idempotent WFSA.
Like the algorithm proposed here, they use \astar~search and on-the-fly determinization;
however, they do not consider decoding over non-idempotent semirings.
We note that the algorithm proposed here could be generalized to compute the $n$ shortest strings over a non-idempotent WFSA.
Specifically, 
one would perform \astar~search 
over the companion semiring using $\beta_d$ as the heuristic as described in \S\ref{s:algorithm},
but would then solve for the $n$ shortest strings \citep[\S6]{Mohri2002}.%
\footnote{
    We thank an anonymous reviewer for this observation.
}

\section{Conclusions}

We propose an algorithm which allows for efficient shortest string decoding of weighted automata over non-idempotent semirings
using \astar~search and on-the-fly determinization.
We find that \astar~search results in a substantial reduction in the number of states visited during decoding,
which in turn minimizes the amount of determinization required to find the shortest string.

We envision several possible applications for the proposed algorithm.
It could be used to exactly decode noisy channel
\say{decipherment} models \citep[e.g.,][]{KnightEtAl2006}
of the form
\begin{equation*}
\hat{P}(p \mid c) \propto P(p) P(c \mid p)
\end{equation*}

\noindent
estimated with ordinary EM,
as well as training scenarios which mix rounds of ordinary and Viterbi EM \citep[e.g.,][]{SpitkovskyEtAl2011}.
The decoding algorithm could also be used for exact decoding of lattices scored with interpolated language models \citep[e.g.,][]{JelinekMercer1980} of the form
\begin{equation*}
\hat{P}(w \mid h) = \lambda_h \tilde{P}(w \mid h) + (1 - \lambda_h) \hat{P}(w \mid h')
\end{equation*}
where $\lambda_h$ is estimated using ordinary EM.

\section{Limitations}

While the evaluation (\S\ref{s:evaluation}) finds the proposed algorithm to be substantially more efficient than the naïve algorithm on real-world data,
it has the same exponential worst-case complexity as exhaustive determinization of acyclic WFSAs.
This worst case dominates the linear-time operations used to compute $\beta_n$ and $\beta_d$, and to solve for the single shortest path.
However, we conjecture that the worst case is unlikely to arise for topologies encountered in actual speech and language processing applications.

\section{Broader impacts}

We are aware of no ethical issues raised by the proposed algorithm beyond issues of dual use, bias, etc.,
which are inherent to all known speech and language technologies.

\section*{Acknowledgments}

Thanks to Michael Riley, Brian Roark, and Jeffrey Sorensen for discussion and code review, and Avital Hirsch and Wen Zhang for proofreading.

\bibliography{astar}

\end{document}